\documentclass[pra,nopacs,twocolumn,10pt]{revtex4}

\usepackage{graphicx,epic,eepic,epsfig,amsmath,latexsym,amssymb,amsmath,verbatim,color}
\usepackage{dsfont}
 \usepackage{hyperref}
\hypersetup{colorlinks=true,citecolor=blue,linkcolor=blue,filecolor=blue,urlcolor=blue,breaklinks=true}
\usepackage{subfigure}
\usepackage{theorem}
\newtheorem{definition}{Definition}
\newtheorem{proposition}[definition]{Proposition}

\newtheorem{theorem}[definition]{Theorem}

\def\squareforqed{\hbox{\rlap{$\sqcap$}$\sqcup$}}
\def\qed{\ifmmode\squareforqed\else{\unskip\nobreak\hfil
\penalty50\hskip1em\null\nobreak\hfil\squareforqed
\parfillskip=0pt\finalhyphendemerits=0\endgraf}\fi}
\def\endenv{\ifmmode\;\else{\unskip\nobreak\hfil
\penalty50\hskip1em\null\nobreak\hfil\;
\parfillskip=0pt\finalhyphendemerits=0\endgraf}\fi}
\newenvironment{proof}{\noindent \textbf{{Proof~} }}{\qed}
\newenvironment{remark}{\noindent \textbf{{Remark~}}}{\qed}

\mathchardef\ordinarycolon\mathcode`\:
\mathcode`\:=\string"8000
\def\vcentcolon{\mathrel{\mathop\ordinarycolon}}
\begingroup \catcode`\:=\active
  \lowercase{\endgroup
  \let :\vcentcolon
  }

\newcommand{\nc}{\newcommand}
\nc{\rnc}{\renewcommand}
\nc{\beg}{\begin{equation}}
\nc{\eeq}{{\end{equation}}}
\nc{\beqa}{\begin{eqnarray}}
\nc{\eeqa}{\end{eqnarray}}
\nc{\lbar}[1]{\overline{#1}}
\nc{\bra}[1]{\langle#1|}
\nc{\ket}[1]{|#1\rangle}
\nc{\ketbra}[2]{|#1\rangle\!\langle#2|}
\nc{\braket}[2]{\langle#1|#2\rangle}

\nc{\proj}[1]{| #1\rangle\!\langle #1 |}
\nc{\avg}[1]{\langle#1\rangle}
\nc{\Rank}{\operatorname{Rank}}
\nc{\smfrac}[2]{\mbox{$\frac{#1}{#2}$}}
\nc{\tr}{\operatorname{Tr}}
\nc{\ox}{\otimes}
\nc{\dg}{\dagger}
\nc{\dn}{\downarrow}
\nc{\cA}{{\cal A}}
\nc{\cB}{{\cal B}}
\nc{\cC}{{\cal C}}
\nc{\cD}{{\cal D}}
\nc{\cE}{{\cal E}}
\nc{\cF}{{\cal F}}
\nc{\cG}{{\cal G}}
\nc{\cH}{{\cal H}}
\nc{\cI}{{\cal I}}
\nc{\cJ}{{\cal J}}
\nc{\cK}{{\cal K}}
\nc{\cL}{{\cal L}}
\nc{\cM}{{\cal M}}
\nc{\cN}{{\cal N}}
\nc{\cO}{{\cal O}}
\nc{\cP}{{\cal P}}
\nc{\cQ}{{\cal Q}}
\nc{\cR}{{\cal R}}
\nc{\cS}{{\cal S}}
\nc{\cT}{{\cal T}}
\nc{\cX}{{\cal X}}
\nc{\cY}{{\cal Y}}
\nc{\cZ}{{\cal Z}}
\nc{\cW}{{\cal W}}
\nc{\csupp}{{\operatorname{csupp}}}
\nc{\qsupp}{{\operatorname{qsupp}}}
\nc{\var}{{\operatorname{var}}}
\nc{\rar}{\rightarrow}
\nc{\lrar}{\longrightarrow}
\nc{\polylog}{{\operatorname{polylog}}}
\nc{\1}{{\mathds{1}}}
\nc{\wt}{{\operatorname{wt}}}
\nc{\av}[1]{{\left\langle {#1} \right\rangle}}
\nc{\supp}{{\operatorname{supp}}}

\def\a{\alpha}

\def\G{\Gamma}

\nc{\RR}{{{\mathbb R}}}
\nc{\CC}{{{\mathbb C}}}
\nc{\FF}{{{\mathbb F}}}
\nc{\NN}{{{\mathbb N}}}
\nc{\ZZ}{{{\mathbb Z}}}
\nc{\PP}{{{\mathbb P}}}
\nc{\QQ}{{{\mathbb Q}}}
\nc{\UU}{{{\mathbb U}}}
\nc{\EE}{{{\mathbb E}}}
\nc{\id}{{\operatorname{id}}}

\nc{\CHSH}{{\operatorname{CHSH}}}

\nc{\be}{\begin{equation}}
\nc{\ee}{{\end{equation}}}
\nc{\bea}{\begin{eqnarray}}
\nc{\eea}{\end{eqnarray}}
\nc{\<}{\langle}
\rnc{\>}{\rangle}
\nc{\Hom}[2]{\mbox{Hom}(\CC^{#1},\CC^{#2})}
\nc{\rU}{\mbox{U}}

\nc{\ob}[1]{#1}

\nc{\SEP}{{\text{SEP}}}
\nc{\NS}{{\text{NS}}}
\nc{\LOCC}{{\text{LOCC}}}
\nc{\PPT}{{\text{PPT}}}
\nc{\EXT}{{\text{EXT}}}
\nc{\Sym}{{\operatorname{Sym}}}

\nc{\ERLO}{{E_{\text{r,LO}}}}
\nc{\ERLOCC}{{E_{\text{r,LOCC}}}}
\nc{\ERPPT}{{E_{\text{r,PPT}}}}
\nc{\ERLOCCinfty}{{E^{\infty}_{\text{r,LOCC}}}}
\nc{\Aram}{{\operatorname{\sf A}}}

\begin{document}
\title{Improved Semidefinite Programming Upper Bound on Distillable Entanglement}
 
\author{Xin Wang$^{1}$}
\email{xin.wang-8@student.uts.edu.au}
\author{Runyao Duan$^{1,2}$}
\email{runyao.duan@uts.edu.au}

\affiliation{$^1$Centre for Quantum Software and Information,\\ Faculty of Engineering and Information Technology,\\ University of Technology Sydney, NSW 2007, Australia}
\affiliation{$^2$UTS-AMSS Joint Research Laboratory for Quantum Computation and Quantum Information Processing, Academy of Mathematics and Systems Science, \\Chinese Academy of Sciences, Beijing 100190, China}

\begin{abstract}
A new additive and semidefinite programming (SDP) computable entanglement measure is introduced to upper bound the amount of distillable entanglement in bipartite quantum states by operations completely preserving the positivity of partial transpose (PPT). This quantity is always smaller than or equal to the logarithmic negativity, the previously best known SDP bound on distillable entanglement, and the inequality is strict in general. Furthermore, a succinct SDP characterization of the one-copy PPT deterministic distillable entanglement for any given state is also obtained, which provides a simple but useful lower bound on the PPT distillable entanglement. Remarkably, there is a genuinely mixed state of which both bounds coincide with the distillable entanglement, while being strictly less than the logarithmic negativity.
\end{abstract}
\maketitle

\textit{Introduction:} 
Entanglement is a striking feature of quantum physics and is a key resource in quantum information processing tasks. A quantitative theory is highly desirable in order to fully exploit the power of entanglement. A series of remarkable  efforts have been devoted both to classifying and quantifying entanglement in the last two decades (for reviews see, e.g., Refs.  \cite{Plenio2005a, Horodecki2009a}). 

One basic entanglement measure is the entanglement of distillation, denoted by $E_D$, which characterizes the rate at which one can obtain maximally entangled states from an entangled state by local operations and classical communication (LOCC) \cite{Bennett1996a, Rains1999a}.  $E_D$ is an important measure because if entanglement is used in a two party protocol, then it is usually required to be in the form of maximally entangled states, e.g., super-dense coding \cite{Bennett1992}
and teleportation \cite{Bennett1993}, and $E_D$ fully captures the ability of a given state to generate \textit{standard} maximally entangled state. Entanglement distillation is also essential for quantum cryptography \cite{Gisin2002} and quantum error-correction \cite{Bennett1996a}. However, how to evaluate $E_D$ for general quantum states remains unknown. 

To quantify bipartite quantum correlations, one of the most popular tools is \emph{negativity} introduced in Ref. \cite{Zyczkowski1998} and it was shown to be an entanglement monotone in Refs. \cite{Vidal2002, Eisert2006, Plenio2005b}.  A more suitable tool is the so-called \emph{logarithmic negativity} $E_N$ \cite{Vidal2002, Plenio2005b}, which remains   the best known semidefinite programming (SDP) computable upper bound on $E_D$ so far  \cite{Rains2001, Vidal2002}. 
 Rains' bound  proposed in Ref. \cite{Rains2001} is the best known upper bound on $E_D$ but recently it is found to be nonadditive \cite{WD16}. Other known upper bounds of $E_D$ have been studied in Refs. \cite{Vedral1998, Rains1999, Horodecki2000a, Christandl2004}. Unfortunately, most of these known entanglement measures are difficult to compute \cite{Huang2014} and usually easily computable only for states with high symmetries, such as Werner states, isotropic states, or the family of ``iso-Werner'' states \cite{Bennett1996a, Vollbrecht2001,Terhal2000, Rains1999}. Thus it is of great interest and significance to find entanglement monotones which are easy to compute for general states.

In this paper we introduce an efficiently computable entanglement measure $E_W$ with an operational interpretation as an improved upper bound on the distillable entanglement, thereby significantly advancing the study of entanglement measures. This quantity is an additive entanglement monotone under both LOCC and a broader class of operations completely preserving the positivity of partial transpose (PPT), and vanishes for the so-called PPT states. For estimating the distillable entanglement, $E_W$ behaves better than the logarithmic negativity. Interestingly, for some states, $E_W$ is equal to the PPT distillable entanglement. With these pleasant properties, $E_W$ is arguably the best known computable and additive entanglement monotone so far. Finally, we obtain an explicit SDP to compute the one-copy PPT deterministic distillable entanglement, which directly provides a computable lower bound of the PPT distillable entanglement.

Before we present our main results, let us first review some notations and preliminaries. In the following we will frequently use symbols such as $A$ (or $A'$) and $B$ (or $B'$) to denote (finite-dimensional) Hilbert spaces associated with Alice and Bob, respectively. The set of linear operators over $A$ is denoted by $\cL(A)$. Note that for a linear operator $R$ over a Hilbert space, we define $|R|=\sqrt{R^\dagger R}$, and the trace norm of $R$ is given by $\|R\|_1=\tr |R|$, where $R^\dagger$ is the conjugate transpose of $R$. The operator norm $\|R\|_\infty$ is defined as the maximum eigenvalue of $|R|$. A deterministic quantum operation (quantum channel) $\cN$ from $A'$ to $B$ is simply a completely positive and trace-preserving (CPTP) linear map from $\cL(A')$ to $\cL(B)$. The Choi-Jamio\l{}kowski matrix of $\cN$ is given by $J_{AB}=\sum_{ij} \ketbra{i_A}{j_{A}} \ox \cN(\ketbra{i_{A'}}{j_{A'}})$, where $\{\ket{i_A}\}$ and $\{\ket{i_{A'}}\}$ are orthonormal basis on isomorphic Hilbert spaces $A$ and $A'$, respectively. 
A positive semidefinite operator $E_{AB} \in \cL(A\ox B)$ is
said to be PPT if $E_{AB}^{T_{B}}\geq 0$, i.e.,
$(\ketbra{i_Aj_B}{k_Al_B})^{T_{B}}=\ketbra{i_Al_B}{k_Aj_B}$.
A bipartite operation $\Pi:\cL(A_i\ox B_i)\rightarrow \cL(A_o\ox B_o)$ is said to be a PPT operation if its Choi-Jamio\l{}kowski matrix is PPT. Separable operations can be defined similarly. A well known fact is that the classes of PPT operations, separable
operations (SEP) \cite{Rains2001} and LOCC obey the following
strict inclusions \cite{Bennett1999b},
\begin{equation}\label{LOCC SEP PPT}
\text{LOCC} \subsetneq  \text{SEP} \subsetneq \text{PPT}.
\end{equation}

The concise definition of entanglement of distillation by LOCC is given in Ref. \cite{Plenio2005a} as follows:
$$E_D(\rho_{AB})=\sup\{r: \lim_{n \to \infty} [\inf_\Lambda  \|\Lambda(\rho_{AB}^{\ox n})- \Phi(2^{rn})\|_1]=0\},$$
where $\Lambda$ ranges over LOCC operations and $\Phi(d)=1/d\sum_{i,j=1}^d\ketbra{ii}{jj}$ represents the standard $d\otimes d$ maximally entangled state.  This can also be used to define the PPT distillable entanglement $E_{\Gamma}(\rho_{AB})$ by replacing LOCC with PPT operations.

In Ref. \cite{Rains2001}, Rains studied entanglement distillation assisted with PPT operations and obtained an upper bound on the distillable entanglement. In deriving this bound, he introduced the ``fidelity of $k$-state PPT distillation'' by
\begin{equation}
    F_{\Gamma}(\rho_{AB},k)
    := \max \{ \tr \Phi(k)
\Pi(\rho_{AB}) :
   \Pi \in \text{PPT}\}，
\end{equation}
which is the optimal entanglement fidelity
of $k \otimes k$ maximally entangled states one can obtain from
$\rho_{AB}$ by PPT operations. Rains simplified $F_{\Gamma}(\rho_{AB},k)$ to
\begin{equation}\begin{split}\label{F PPT}
   F_{\Gamma}(\rho_{AB}&,k)= \max   \tr \rho_{AB}Q_{AB}, \\ 
   \text{ s.t. }\   &0\le Q_{AB} \le \1, 
   -\frac{1}{k}\1 \le Q_{AB}^{T_{B}} \le  \frac{1}{k}\1.
\end{split}\end{equation}
And the PPT distillable entanglement can be equivalently defined as
\begin{equation}
    E_{\Gamma}(\rho_{AB}):=
    \sup \{ r : \lim_{n \to \infty}
      F_{\Gamma}(\rho_{AB}^{\ox n}, 2^{nr}) = 1 \}. 
\end{equation}

The \emph{logarithmic negativity} of a state $\rho_{AB}$ mentioned above is defined as \cite{Vidal2002, Plenio2005b} 
\begin{equation}\label{EN}
E_N(\rho_{AB})=\log_2 \|\rho_{AB}^{T_B}\|_1.
\end{equation}
As shown in Refs. \cite{Rains2001, Vidal2002}, the significance of $E_N$ is highlighted in the following
$$E_D(\rho_{AB})\le  E_{\Gamma}(\rho_{AB}) \le E_{N}(\rho_{AB}).$$

The entanglement monotone is one of the most essential features for a function to quantify the entanglement.
Any (non-negative) function $E(\cdot)$ over bipartite states is said to be an entanglement monotone if it does not increase on average under general LOCC  (or PPT) operations \cite{Plenio2005b}, i.e.,
\begin{equation}\label{E monotone}
E(\rho)\ge \sum_ip_iE(\rho_i),
\end{equation}
where state $\rho_i$ with label $i$ is obtained with probability $p_i$ in the LOCC (PPT) protocol applied to $\rho$. 

SDP problems \cite{Vandenberghe1996} can be solved by polynomial time algorithms \cite{Khachiyan1980}. The CVX software  \cite{Grant2008} allows one to solve SDPs efficiently. More details on this topic can be found in Ref. \cite{Watrous2011b}. Here clearly $F_{\Gamma}(\rho_{AB},k)$ is SDP computable for any state $\rho_{AB}$ and positive real number $k$ (not necessary to be integers). However, it remains unclear whether $E_{\Gamma}(\rho_{AB})$ is also SDP computable due to the complicated limiting procedure in the definition. Interestingly, for any bipartite pure state $E_\G$ coincides with the entropy of entanglement \cite{Matthews2008}.


\textit{A new SDP upper bound on distillable entanglement:}  
We are now ready to introduce an SDP upper bound $E_W$ on $E_{\Gamma}$ and thus also on $E_D$, as follows: 
$$E_W(\rho_{AB})=\log_2 W(\rho_{AB}),$$ where $W(\rho_{AB})$  is given by the following SDP:
\begin{equation}\label{prime WN}
\begin{split}
W(\rho_{AB})= \max   &\tr  \rho_{AB}^{T_B}R_{AB}, \\
\phantom{ W(\rho) }\text{ s.t. }\   & \|R_{AB}\|_\infty \le  1, R_{AB}^{T_B}\ge0. 
\end{split}\end{equation}
Noticing that the constraint $\|R_{AB}\|_\infty \le  1$ can be rewritten as $-\1\le R_{AB} \le  \1$, we can use the Lagrange multiplier approach to obtain the dual SDP as follows:
\begin{equation}\label{dual WN}
\begin{split}
W(\rho_{AB})&= \min   \tr(U_{AB}+V_{AB}), \\
\phantom{ W(\rho) }  \text{ s.t. }\   &U_{AB},V_{AB}\ge0, (U_{AB}-V_{AB})^{T_{B}}  \ge \rho_{AB}.
\end{split}\end{equation}
It is worth noting that the optimal values of the primal and the dual SDPs above coincide. This is a consequence of strong duality. By Slater's condition \cite{Slater2014}, one simply needs to show that there exists positive definite $U_{AB}$ and $V_{AB}$ such that $(U_{AB}-V_{AB})^{T_{B}} > \rho_{AB}$, which holds for $U_{AB}=3V_{AB}=3\1$.
Introducing a new variable operator $X_{AB}=(U_{AB}-V_{AB})^{T_B}$, we can further simplify the dual SDP to
\begin{equation}\label{dual 2 WN}
W(\rho_{AB})= \min \|X_{AB}^{T_B}\|_1, \ \text{ s.t. }\    X_{AB}  \ge \rho_{AB}.
\end{equation}

The function $E_W(\cdot)$ has the following remarkable properties which will be discussed in greater details shortly:
\begin{itemize}
\item[i)] Additivity under tensor product: $E_W(\rho_{AB}\ox \sigma_{A'B'})=E_W(\rho_{AB})+E_W(\sigma_{A'B'})$.
\item[ii)] Upper bound on PPT distillable entanglement: $E_{\Gamma}(\rho_{AB})\le E_W(\rho_{AB})$. 
\item[iii)] Detecting genuine PPT distillable entanglement: $E_W(\rho_{AB})>0$ if and only if $\rho_{AB}$ is PPT distillable.
\item[iv)] Entanglement monotone under general LOCC (or PPT) operations:
$E_W(\rho)\ge \sum_ip_iE_W(\rho_i)$.
\item [v)] Improved bound over logarithmic negativity: $E_W(\rho_{AB})\le E_N(\rho_{AB})$, and the inequality can be strict.
\end{itemize}

In the rest of this section we will focus on properties i) to iii). Properties iv) and v) will be discussed in the subsequent sections.

Property i) is equivalent to the multiplicativity of the function $W(\cdot)$ under tensor product and can be proven directly by using the primal and dual SDPs of $W(\cdot)$. To see the super-multiplicativity, suppose that the optimal solutions to the primal SDP (\ref{prime WN}) of $W(\rho_{AB})$ and $W(\sigma_{A'B'})$ are  $R_{AB}$ and $S_{A'B'}$, respectively.  
We need to show that $R_{AB}\ox S_{A'B'}$ is a feasible solution to the primal  SDP (\ref{prime WN}) of $W(\rho_{AB} \ox \sigma_{A'B'})$. That will imply $W(\rho_{AB} \ox \sigma_{A'B'})\ge \tr (\rho_{AB}^{T_B}\ox \sigma_{A'B'}^{T_{B'}})(R_{AB}\ox S_{A'B'})= W(\rho_{AB})W(\sigma_{A'B'})$. The proof is quite straightforward. Indeed from  $\|R_{AB}\|_\infty\leq 1$ and $\|S_{A'B'}\|_\infty\leq 1$, $\|R_{AB}\ox S_{A'B'}\|_\infty\leq 1$ follows immediately. Also the positivity of $R_{AB}^{T_B}\otimes S_{A'B'}^{T_{B'}}$ is obvious. Hence we are done. The sub-multiplicativity of $W(\cdot)$ can be proven similarly refer to the dual SDP (\ref{dual 2 WN}) of $W(\rho_{AB})$.

Property ii) requires some effort and is presented in the following
\begin{theorem}\label{DPPT EW}
For any state $\rho_{AB}$, $E_{\Gamma}(\rho_{AB})\le E_W(\rho_{AB}).$
\end{theorem}
\begin{proof}
Suppose $E_{\Gamma}(\rho_{AB})=r$. Then $$\lim_{n \to \infty}      F_{\Gamma}(\rho_{AB}^{\ox n}, 2^{nr}) = 1.$$

For a given $k$, suppose that the optimal solution to the SDP (\ref{F PPT}) of $F_{\Gamma}(\rho_{AB}, k)$ is $Q_{AB}$. 
Let $R_{AB}=kQ_{AB}^{T_{B}}$. Then from the constraints of SDP (\ref{F PPT}), 
we have that $-\1 \le R_{AB}=kQ_{AB}^{T_B} \le  \1$. It is also clear that $R_{AB}^{T_{B}}\ge 0$. So $R_{AB}$ is a feasible solution to the primal SDP (\ref{prime WN}) of $W(\rho_{AB})$. Therefore,  
$$W(\rho_{AB}) \ge \tr \rho_{AB}^{T_{B}}R_{AB}=k\tr \rho_{AB}Q_{AB}=kF_{\Gamma}(\rho_{AB}, k).$$
Hence, $$\lim_{n \to \infty}   {W(\rho_{AB}^{\ox n})}/{2^{nr}}\ge\lim_{n \to \infty}     F_{\Gamma}(\rho_{AB}^{\ox n}, 2^{nr}) =1.$$
Noticing that $W(\rho)$ is multiplicative, we have
$$\lim_{n \to \infty}    {{W(\rho_{AB}^{\ox n})}}/{2^{nr}}=\lim_{n \to \infty}   {(W(\rho_{AB}))^n}/{2^{nr}}\ge1.$$
Therefore, $W(\rho_{AB}) \ge 2^r$, and we are done.
\end{proof}

Property iii) suggests an interesting equivalent relation between $E_W$ and $E_\G$ in the sense that $E_W$ can be used to detect whether a state is genuinely distillable under PPT operations.
\begin{theorem}
For a state $\rho_{AB}$, $E_W(\rho_{AB})>0$ if and only if $E_\G(\rho_{AB})>0$.
\end{theorem}
\begin{proof}
We only need to show that $W(\rho_{AB})>1$ is equivalent to $\rho_{AB}$ is an non-positive partial transpose (NPPT) state. The rest of the  proof then can be completed by combining this fact with an interesting result  from Ref. \cite{Eggeling2001}: Any NPPT state is PPT distillable.

Firstly, if $\rho_{AB}$ is PPT, then $W(\rho_{AB})\le \|\rho_{AB}^{T_B} \|_1=1$.  Assume now $\rho_{AB}$ is NPPT, we will show that $W(\rho_{AB})>1$.  Let $P_{-}$ be the projection on the subspace spanned by the eigenvectors with negative eigenvalues of $\rho_{AB}^{T_B}$, and let $\lambda=\|P_{-}^{T_B}\|_\infty$. Introduce 
$$R_{AB}=\1_{AB}-\frac{1}{\max\{\lambda,0.5\}}P_{-}.$$ 
It is clear that $R_{AB}^{T_B}\ge 0$ by construction.  Furthermore, we can easily verify that $-\1\le\1-2P_{-}\le R_{AB}\le \1$. So $R_{AB}$  is a feasible solution to the primal  SDP (\ref{prime WN}) of $W(\rho_{AB})$. Noticing that $\rho_{AB}$ is NPPT, we have that
\begin{align*}
W(\rho_{AB}) &\ge\tr \rho_{AB}^{T_B} R_{AB}
=1-\frac{\tr P_{-}\rho_{AB}^{T_B}}{\max\{\lambda,0.5\}}>1,
\end{align*}
where we have used the property that $\tr P_{-}\rho_{AB}^{T_B}<0$.
\end{proof}

\textit{$E_W$ is an entanglement monotone:} 
We are going to prove that $E_W$  is a proper entanglement monotone in the sense of Eq. (\ref{E monotone})  under general PPT operations, and then it implies monotonicity for LOCC. The approach is in the spirit of the proof of the monotonicity of  logarithmic negativity in Ref. \cite{Plenio2005b}.
\begin{theorem}
The function $E_W(\cdot)$ is an entanglement monotone both under general LOCC and PPT operations. 
\end{theorem}
\begin{proof}
Let us consider a general PPT operation $\cN=\sum_i \cN_i$ that maps the bipartite state $\rho$ to $\cN_i(\rho)/\tr(\cN_i(\rho))$ with probability $\tr \cN_i(\rho)$, where $\cN_i$ is CP and PPT operation.

We suppose that $X_{AB}$  is the optimal solution to the dual SDP (\ref{dual 2 WN}) of $W(\rho_{AB})$. It is easy to see that $\cN_i(X_{AB})\ge \cN_i(\rho)$, then $\cN_i(X_{AB})$ is feasible to the dual SDP (\ref{dual 2 WN})  of $W(\cN_i(\rho))$. Therefore, 
$$W(\cN_i(\rho))\le \|(\cN_i(X_{AB}))^{T_B}\|_1=\tr |\cN_i^{T_B}(X_{AB}^{T_B})|,$$
where $\cN_i^{T_B}(\sigma)=(\cN(\sigma^{T_B}))^{T_B}$. By the  fact that $\cN_i^{T_B}$ is CP  \cite{Rains1999, Rains2001}, we have 
$W(\cN_i(\rho))\le \tr |\cN_i^{T_B}(X_{AB}^{T_B})|
\le \tr \cN_i^{T_B}(|X_{AB}^{T_B}|)$.
Then, we have that
\begin{align*}
\sum_i p_iE_W(\rho_i)&\le \log_2\sum_ip_iW(\rho_i)
=\log_2\sum_iW(\cN_i(\rho))\\
&\le  \log_2\sum_i \tr \cN_i^{T_B}(|X_{AB}^{T_B}|)\\
&=\log_2\sum_i \tr [\cN_i(|X_{AB}^{T_B}|^{T_B})]^{T_B}\\
&=\log_2 \tr \cN(|X_{AB}^{T_B}|^{T_B})\\
&=\log_2 \tr |X_{AB}^{T_B}|^{T_B}
=E_W(\rho).
\end{align*}

Hence, we obtain the monotonicity of $E_W$ under general PPT operations in the sense of Eq. (\ref{E monotone}). Similar to the logarithmic negativity, one can easily conclude that $E_W$ is also a full non-convex entanglement monotone.
\end{proof}

\textit{Comparison with logarithmic negativity:}
Now we discuss property iv). 
Before that, let us recall that $\|\rho_{AB}^{T_B}\|_1$ can be reformulated as 
\begin{equation}\label{trace norm}
\|\rho_{AB}^{T_B}\|_1= \max   \tr  \rho_{AB}^{T_B}R_{AB} \ \text{ s.t. }\    \|R_{AB} \|_\infty\le 1.
\end{equation}
\begin{theorem}
For any state $\rho_{AB}$,
$E_W(\rho_{AB}) \le E_N(\rho_{AB}),$
and the inequality can be strict. Moreover, $E_W(\rho_{AB})=E_N(\rho_{AB})$ if and only if SDP (\ref{trace norm}) has an optimal solution with positive partial transpose.
\end{theorem}
\begin{proof}
The definition of $E_N$ is given in Eq. (\ref{EN}).
Noting that $\rho_{AB}$ is a feasible solution to the dual SDP (\ref{dual 2 WN}) of $W(\rho_{AB}) $, we have
$E_W(\rho_{AB})\le \log \|\rho_{AB}^{T_B}\|_1= E_N(\rho_{AB})$.

To see  the above inequality can be strict, we focus on a class of two-qubit states 
$\sigma^{(r)}=r\proj{v_0}+(1-r)\proj{v_1}$ $(0<r<1)$, 
where $\ket{v_0}={1}/{\sqrt 2}(\ket {10}-\ket{11})$ and $\ket{v_1}={1}/{\sqrt 3}(\ket {00}+\ket{10}+\ket{11})$. The fact that $E_W(\sigma^{(r)})$ can be strictly smaller than $E_N(\sigma^{(r)})$ is shown in FIG. \ref{comp1}.

To prove the second part of the theorem, let us assume that the optimal solution to SDP (\ref{trace norm}) of $\|\rho_{AB}^{T_B}\|_1$ is $R_{AB}$. If $R_{AB}^{T_B}\ge0$, then it is also a feasible solution to the primal SDP (\ref{prime WN}) of $W(\rho_{AB})$. That immediately implies $E_W(\rho_{AB})=E_N(\rho_{AB})$. Conversely, assume that $E_W(\rho_{AB})=E_N(\rho_{AB})$, then the optimal solution $R_{AB}$ to SDP (\ref{prime WN}) of $W(\rho_{AB})$ is also the optimal solution to the SDP (\ref{trace norm}) for $\|\rho_{AB}^{T_B}\|_1$ and it holds that $R_{AB}^{T_B}\ge0$.  Therefore, $E_W(\rho_{AB})= E_N(\rho_{AB})$ if and only if  SDP (\ref{trace norm}) for $\|\rho_{AB}^{T_B}\|_1$ has a PPT optimal solution.
\begin{figure}[h]
\subfigure{\includegraphics[width=4.9cm]{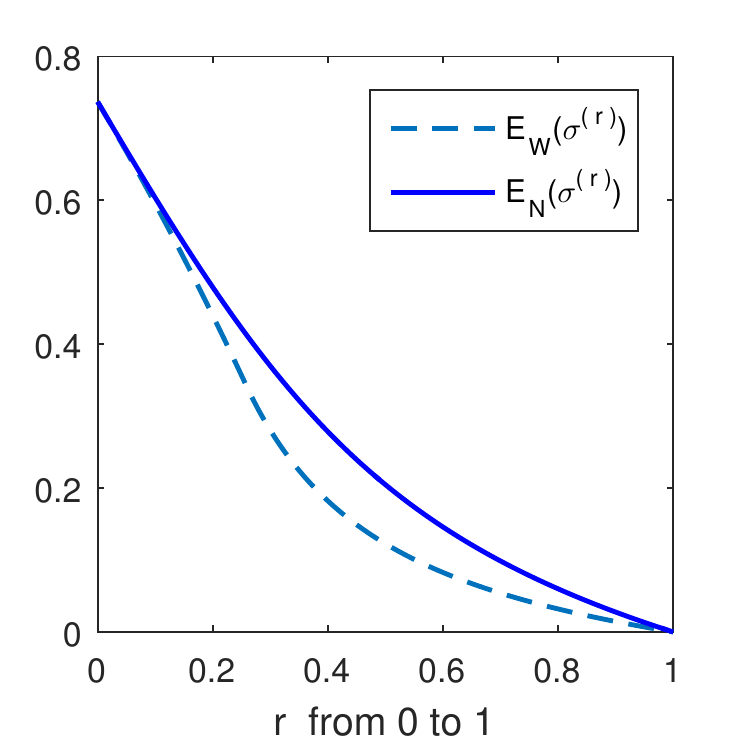}}
\caption{This plot presents the comparison of different upper bounds on $E_\G(\sigma^{(r)})$ (and $E_D(\sigma^{(r)})$ as well). The dashed line depicts  $E_W(\sigma^{(r)})$ while the solid line depicts $E_N(\sigma^{(r)})$.}
\label{comp1}
\end{figure}
\end{proof}

We further compare $E_W$ to $E_\G$ and $E_N$ by using a class of $3\otimes 3$ states defined by 
$$\rho^{(\a)}=\sum_{m=0}^{2}\proj{\psi_m}/3 \ (0<\a \le 0.5)$$
 with $\ket{\psi_0}=\sqrt{\a}\ket{01}+\sqrt{1-\a}\ket{10}$,   $\ket{\psi_1}=\sqrt{\a}\ket{02}+\sqrt{1-\a}\ket{20}$ and $\ket{\psi_2}=\sqrt{\a}\ket{12}+\sqrt{1-\a}\ket{21}$.  

\begin{proposition}\label{EW tight}
For the class of states $\rho^{(\a)}$, we have that
$$E_{\Gamma}(\rho^{(\a)}) \le E_W(\rho^{(\a)}) < E_N(\rho^{(\a)}).$$
In particular, 
$$E_{\Gamma}(\rho^{(0.5)}) = E_W(\rho^{(0.5)})=\log_2 \frac{3}{2} < \log_2 \frac{5}{3}=E_N(\rho^{(0.5)}).$$
\end{proposition}
\begin{proof}
Firstly, 
it is not difficult to see that
\begin{equation} \label{eq EN}
E_N(\rho^{(\a)})=\log_2\|(\rho^{(\a)})^{T_B} \|_1=  \log_2(1+{4}/{3}\sqrt{\a(1-\a)}). 
\end{equation}

Secondly, we can choose $X_{AB}=\rho^{(\a)}+{\sqrt{\a(1-\a)}}/{3}(\proj{00}+\proj{11}+\proj{22})$ as a feasible solution to the dual SDP (\ref {dual 2 WN}). By a routine calculation, we have
\begin{equation}\label{eq EW}
\begin{split}
E_W(\rho^{(\a)}) &=\log_2 W(\rho^{(\a)}) \le \log_2 \|X_{AB}^{T_B}\|_1\\
&=\log_2 (1+\sqrt{\a(1-\a)})<E_N(\rho^{(\a)}).
\end{split} \end{equation}

For $\a=0.5$, choose $k_0={3}/{2}$ and $Q_{AB}=\sum_{m=0}^{2}( \proj{\psi_m}+{1}/{3}\proj{\widehat\psi_m})$ with 
$\ket {\widehat \psi_0} =  1/\sqrt{2}(\ket{01}  -\ket{10})$, $\ket {\widehat \psi_1} =  1/\sqrt{2}(\ket{02}  -\ket{20})$ and $\ket {\widehat \psi_2} =  1/\sqrt{2}(\ket{12}  -\ket{21})$.
Noticing that $\|Q_{AB}^{T_B}\|_\infty={2}/{3}$, we have $-{1}/{k_0}\1 \le Q_{AB}^{T_{B}} \le  {1}/{k_0}\1$. Thus $Q_{AB}$ is a feasible solution to the SDP (\ref{F PPT}) of $F_{\Gamma}(\rho^{(0.5)},k_0)$, which has an optimal value $1$ due to  
$1\geq F_{\Gamma}(\rho^{(0.5)},k_0)\geq \tr  \rho^{(0.5)}Q_{AB}=1$.
Applying the definition of $E_\G$, we have
\begin{equation}  \label{eq D}
E_{\Gamma}(\rho^{(0.5)}) \ge\log_2 k_0=\log_2{3}/{2}.
\end{equation}

Finally, combining Eqs. (\ref  {eq EN}), (\ref {eq EW}), and (\ref {eq D}), we obtain the desired chain of inequalities. 
\end{proof}

\begin{remark}
It is worth pointing out that $\rho^{(0.5)}$ is supporting on the symmetric subspace $\rm{span}\{\ket{01}+\ket{10}, \ket{02}+\ket{20},\ket{12}+\ket{21}\}$, which looks quite similar to but is  actually not locally unitarily equivalent to the antisymmetric subspace $\rm{span}\{\ket{01}-\ket{10}, \ket{02}-\ket{20},\ket{12}-\ket{21}\}$. In particular, for the corresponding $3\ox 3$ antisymmetric state $\sigma_3$, we have $E_\G(\sigma_3)=E_W(\sigma_3)=E_N(\sigma_3)=\log_2(5/3)$. 
 \end{remark}

\textit{PPT deterministic distillable entanglement:}
The deterministic entanglement distillation concerns about how to distill maximally entangled states exactly.  The bipartite pure state case is completely solved in Refs. \cite{Matthews2008, DFJY2004}. We will show that PPT deterministic distillable entanglement of a state $\rho_{AB}$ depends only on the support $\supp(\rho_{AB})$, which is defined to be the space spanned by the eigenvectors with positive eigenvalues of $\rho_{AB}$. The one-copy PPT deterministic distillable entanglement of  $\rho_{AB}$ is defined by 
$$
        E_{\Gamma,0}^{(1)}(\rho_{AB})
        := \max \left\{  \log_2 k
        : F_\Gamma(\rho_{AB}, k) = 1, k >0 \right\}.
$$
Clearly $E_{\Gamma,0}^{(1)}(\rho_{AB})\geq 0$ since $F_\Gamma(\rho_{AB}, 1)=1$ trivially holds. The asymptotic PPT deterministic distillable entanglemen of $\rho$ is given by
$$E_{\Gamma,0}(\rho_{AB}):=\sup_{n\ge 1}{E_{\Gamma,0}^{(1)}(\rho^{\otimes n})}/{n}=\lim_{n\ge 1}{E_{\Gamma,0}^{(1)}(\rho_{AB}^{\otimes n})}/{n}.$$ 

Replacing $k$ and $Q_{AB}$ in SDP (\ref{F PPT}) by $\tr  \rho_{AB} R_{AB}$ and $R_{AB}/\tr  \rho_{AB} R_{AB}$, respectively, we can further simplify 
$E_{\Gamma,0}^{(1)}(\rho_{AB})$ to $\log_2 W_0({\rho_{AB}})$ such that
\begin{equation}\label{prime W0}
\begin{split}
W_0(\rho_{AB})&= \max   \tr  \rho_{AB} R_{AB}, \\
\phantom{W(\rho) }\text{ s.t. }  &  0\le R_{AB}\le (\tr  \rho_{AB} R_{AB})\1_{AB},\\ 
\phantom{W(\rho) }  & |R_{AB}^{T_{B}}| \le  \1_{AB}.
\end{split}\end{equation}
The first constraint in SDP (\ref{prime W0}) implies that $\tr  \rho_{AB} R_{AB}\ge \|R_{AB}\|_\infty$. So any feasible $R_{AB}$ should be of the form $xP_{AB}+S_{AB}$, where $x\geq 0$, $P_{AB}$ is the projection onto $\supp(\rho_{AB})$, and $0\leq S_{AB}\leq x(\1-P)_{AB}$.  
Replacing $S_{AB}/x+P_{AB}$ by $R_{AB}$ and noticing $E_{\Gamma,0}^{(1)}(\rho_{AB})=\log_2 W_0(\rho_{AB})$, we have 
\begin{equation}\label{prime NW1}
\begin{split}
E_{\Gamma,0}^{(1)}(\rho_{AB})&= \max_R -\log_2 \|R_{AB}^{T_B}\|_\infty, \\
\phantom{W(\rho) }\text{ s.t. }  &  P_{AB}\le R_{AB}\le \1_{AB}.
\end{split}
\end{equation}
In particular, $E_{\Gamma,0}^{(1)}(\rho_{AB})\geq -\log_2 \|P_{AB}^{T_B}\|_\infty$ when $R_{AB}=P_{AB}$. For bipartite pure entangled states this lower bound gives the exact value of the PPT deterministic distillable entanglement \cite{Matthews2008, DFJY2004}. However, this is not the case for general mixed bipartite states such as $\rho^{(0.5)}$. Clearly we have 
$$E_{\Gamma,0}^{(1)}\leq E_{\Gamma,0}\leq E_{\G}\leq E_{W}\leq E_{N},$$ and for $\rho^{(0.5)}$, and the first three inequalities become an equality while the last one is strict. 
Recently, the SDP (\ref{prime NW1}) of $E_{\Gamma,0}^{(1)}$ was used to evaluate the PPT distillable entanglement of the rank-$2$ antisymmetric state \cite{Wang2016d}.

\textit{Conclusions:}
We present a new and improved SDP-computable upper bound $E_W$ to the distillable entanglement. This quantity enjoys additional good properties such as additivity and monotonicity under both general LOCC  (or PPT) operations. $E_W$ has almost all of good properties of logarithmic negativity  while can provide a more accurate estimation of the distillable entanglement and it has been recently used to give a SDP-computable sufficient condition of the irreversibility of asymptotic entanglement manipulation under PPT operations \cite{Wang2016d}.
We also show that the PPT deterministic distillable entanglement depends only on the support of the state and provides a refined SDP for the one-copy rate, which is a natural lower bound of the PPT distillable entanglement. 

One interesting open problem is whether $E_W(\rho^{(0.5)})$ in Proposition \ref{EW tight} is achievable by LOCC.  We hope that this SDP-computable entanglement measure would be useful in studying other problems in quantum information theory.

We were grateful to A. Winter and Y. Huang for helpful suggestions and M. Plenio and J. Eisert for communicating references to us. This work was partly supported by the Australian Research Council (Grant Nos. DP120103776 and FT120100449) and the National Natural Science Foundation of China (Grant No. 61179030).


\end{document}